\documentclass[12pt,a4paper,twoside,reqno]{amsart}
\usepackage[english]{babel}
\usepackage{amssymb,graphicx,color}
\usepackage{mathrsfs}
\usepackage{caption}
\usepackage[left=2cm,right=2cm,
top=35mm,bottom=34mm,bindingoffset=0cm]{geometry}
\usepackage[colorlinks=true, linktocpage=false,
            pdfwindowui=false, pdfmenubar=false, 
            citecolor=blue,
            hyperfootnotes=true, pdfstartview=FitH,
            bookmarksopen=false, bookmarks=false]{hyperref}

\DeclareMathOperator*{\Ccup}{\text{\large$\cup$}}
\def~{\unskip\nobreak\kern0.2778em\ignorespaces\nobreak\hskip0pt}

\def\sss{\scriptscriptstyle}
\def\bo{\boldsymbol}

\def\blue{\color{blue}}

\def\DEF{ \mathrel{\mathop:}=}
\def\FED{=\mathrel{\mathop:} }
\def\scr{\mathscr}
\def\const{\mathrm{const}}
\def\ie{i.\kern0.1667em e.\@}

\def\bop{\bo{\frak p}}
\def\po{\frak p_\circ}
\def\frak{\mathfrak}
\def\pa{\partial}
\def\ri{{\mathrm i}}

\newtheorem{proposition}{Proposition}
\newtheorem{remark}{Remark}
\newtheorem*{definition}{Definition}

\newmuskip\Ymuskip \Ymuskip = 1mu
\renewcommand{\,}[1][1]{%
\ifmmode\mskip#1\Ymuskip\else\kern#1\dimexpr0.05555em\relax\fi}
\renewcommand{\!}[1][1]{\,[-#1]}

\newcommand{\goto}%
{\mathrel{{\mathchar"439}{\mathchar"439}\mathchar"044B}}

\newcommand{\mbig}[2][4]{\vcenter{\hbox{%
\ifcase#1$#2$\or\small$\big#2$\or$\big#2$\or\large$\big#2$%
\or\footnotesize$\Big#2$\or\small$\Big#2$\or$\Big#2$\or%
\large$\Big#2$\or$\bigg#2$\or\large$\bigg#2$\or$\Bigg#2$\or%
\large$\Bigg#2$\or$\left#2\vbox to 4.8ex{}\right.$\else%
${\left#2\vbox to 5.3ex{}\right.}$\fi}}\relax}

\title[Exactly solvable wave fronts]%
{An exactly solvable problem of wave fronts\\[0.5ex]
and applications to the asymptotic theory}

\author{Yu.~V.~Brezhnev}
\author{A.~V.~Tsvetkova}

\begin{document}
\maketitle
\thispagestyle{empty}

\begin{abstract}
This is the full and extended version of the brief note
\href{https://arxiv.org/abs/1908.00938}
{\blue\texttt{arxiv.org/abs/1908.00938}}. A nontrivially solvable
4-dimensional Hamiltonian system is applied to the problem of wave
fronts and to the asymptotic theory of partial differential
equations. The Hamilton function we consider is $H(\mathbf x,\mathbf
p)=\sqrt{D(\mathbf{x})}|\mathbf{p}|$. Such Hamiltonians arise when
describing the fronts of linear waves generated by a localized
source in a basin with a variable depth. We consider two
\emph{realistic} types of bottom shape: 1) the depth of the basin is
determined, in the polar coordinates, by the function
$D(\varrho,\varphi)=(\varrho^2+b)/(\varrho^2+a)$ and 2) the depth
function is $D(x,y)=(x^2+b)/(x^2+a)$. As an application, we
construct the asymptotic solution to the wave equation with
localized initial conditions and asymptotic solutions of the
Helmholtz equation with a localized right-hand side.
\end{abstract}

\tableofcontents

\section*{Introduction}

In the present paper we consider an exactly solvable Hamiltonian
system
\begin{equation}\label{Ham_syst}
\dot{\mathbf{x}}=H_\mathbf{p}\,,\quad\dot{\mathbf{p}}=-H_\mathbf{x}\,,\quad
\mathbf x,\mathbf p\in\mathbb R^2
\end{equation}
with the Hamiltonian
\begin{equation}\label{Ham}
H=|\mathbf{p}|C(\mathbf{x}),\quad C(\mathbf{x})=\sqrt{g D(\mathbf{x})}
\end{equation}
and an associated problem of the wave fronts under the initial
conditions
\begin{equation}\label{cond}
\mathbf{p}|_{t=t_\circ}=\mathbf{n}(\psi)=\begin{pmatrix} \cos\psi\\
\sin\psi\end{pmatrix}, \quad \psi\in[0,2\pi],\quad
\mathbf{x}|_{t=t_\circ}=\mathbf{x}^{\circ}\in \mathbb{R}^2\,.
\end{equation}
The function $D(\mathbf{x})$ is assumed to be one of two types,
which will be given bellow. Solutions to these problems can then be
represented in terms of elliptic functions \cite{Ahiezer}.

Among other things the Hamiltonian systems of that sort arise in the
study of waves generated by a localized source. This problem, in
turn, appears in modeling the tsunami waves \cite{Peli} and
mesoscale eddies in the ocean. It is well known \cite{Stoker, Peli}
that the long-wave approximation in an infinite basin with a
variable depth defined by the function $D(\mathbf{x})>0$ leads to
the linear wave equation of the form
\begin{equation}\label{WEq1}
\frac{\pa^2 w}{\pa t^2}- \big\langle \nabla,
C^2(\mathbf{x})\nabla \big\rangle w=0,\quad
C^2(\mathbf{x})=gD(\mathbf{x}),\quad \mathbf{x}\in \mathbb{R}^2\,.
\end{equation}
Here, $g$ is the gravity acceleration, $\nabla$ denotes the
two-dimensional gradient, $\langle\cdot,\cdot\rangle$ is the usual
scalar product, so $\big\langle \nabla, C^2(\mathbf{x})\nabla
\big\rangle w=C^2(\mathbf{x})\Delta w
+2C(\mathbf{x})(C'_{x_1}w'_{x_1}+C'_{x_2}w'_{x_2})$. For this
equation we consider the Cauchy problem with localized initial
conditions
\begin{equation}\label{WEq2}
w|_{t=0}=v\mbig[7](\frac{\mathbf{x}-\mathbf{x}^\circ}{l}\mbig[7]),\quad
\frac{\pa w}{\pa t}\Big|_{t=0}=0,
\end{equation}
where smooth function $V(y)$ (that decreasing fast as $|y|\to
\infty$) characterizes the uplift of the ocean surface, the
parameter $l$ characterizes the source size, and
$\mathbf{x}^{\circ}$ is a point in whose neighborhood the initial
perturbation is localized.

To the best of our knowledge, the complete analytic solutions to
problems of this kind are either absent or are somewhat
artificial/non-realistic. On the other hand, the realistic
Hamiltonians and depth-functions may lead, when solvable, to
nontrivial representation problems with (elliptic) theta-functions,
especially considering the fact that the ultimate formulas may give
rise to the rather nonstandard inversion problems \cite{alber2}. It
is this situation---transcendental equations in
sect.~\ref{case}---that arises in the procedure of integrating the
system \eqref{Ham_syst}--\eqref{Ham}. In the theory of integrable
systems \cite{fritz, dickey, belokolos}, such equations correspond
to a non-linear evolution on Jacobians and to inversion of
meromorphic \cite{alber1, alber2} and logarithmic integrals
\emph{rather than} the holomorphic ones. See also
\cite[sect.~9.2]{br1} for explicit formulas and monograph
\cite{fritz} for extensive bibliography along these lines.

 It is also known that solvability in terms of elliptic
functions \cite{Ahiezer} constitutes presently a school in its own
right and impart the great analytic effectiveness to the general
integrability-theory \cite{belokolos, dickey}; even the `modular
part' of the elliptic theory---Weierstrass' parameters
$g_2,g_3$---meets (nonstandard) inversion problems and has
nontrivial applications \cite{br2}.

As for the wave applications, we consider waves propagating over
underwater banks and ridges. Accordingly, the function $D$ that
determines the shape of the basin bottom has one of two types. In
the case of a bank, the basin depth is defined by the function
\begin{equation}\label{D1}
D(\varrho,\varphi)=\frac{\varrho^2+b}{\varrho^2+a}\,
\end{equation}
where $\varrho$ is a polar radius and $a>b>0$ are constants. For a
ridge, the basin depth is defined by the function
\begin{equation}\label{D2}
D(\mathbf{x})=\frac{x_1^2+b}{x_1^2+a}, \quad \mathbf{x}=(x_1,x_2)
\end{equation}
with the same restriction on the constants $a$ and $b$. The exact
analytical solutions of the corresponding Hamiltonian systems are
given in sect.~\ref{case1} and sect.~\ref{case2} respectively.

The problem of constructing such solutions is interesting both in
its own right and in applications. In particular, the system
\eqref{Ham_syst}--\eqref{Ham} arises when constructing the
asymptotic solutions to the various boundary-value problems for the
wave equation with a variable velocity $C(\mathbf{x})$ by
application of the well-known in physics \emph{the ray method}
\cite[sect.~I]{BabicBul}. This technique is based on the eikonal
equation---the Hamilton--Jacobi equation for a Hamiltonian---and on
the transport equations through a transition to the ray coordinates.
 Let us briefly describe the idea of the method and explain
how the asymptotic solution of equation \eqref{WEq1} is related to
the Hamiltonian~\eqref{Ham}.

If the source-localization parameter $l$ tends to zero we may seek
an asymptotic solution of the problem \eqref{WEq1}--\eqref{WEq2} in
the form
\begin{equation*}
w=f\mbig[7](\frac{t-s(\mathbf{x})}{l},\mathbf{x},t\mbig[7])+O(l)\,,
\end{equation*}
where $f(y,\mathbf{x})$ and $s(\mathbf{x})$ are smooth functions,
and $f(y,\mathbf{x})$ is decreasing fast as $|y| \rightarrow
\infty$. Substituting this ansatz into \eqref{WEq1} and equating the
coefficient at $l^{-2}$ to zero, we get the Hamilton--Jacobi
equation for the function $s(\mathbf{x})$:
\begin{equation}\label{HJeq}
C(\mathbf{x})\sqrt{\left(\frac{\partial s(\mathbf{x})}{\partial
x_1}\right)^2+\left(\frac{\partial s(\mathbf{x})}{\partial
x_2}\right)^2}=\pm 1.
\end{equation}
Note that the left-hand side of this equation is obtained from
function \eqref{Ham} by replacing $p_i$ with $\frac{\partial
s}{\partial x_i}$ ($i=1,2$).  Equation \eqref{HJeq} can thus be
solved along characteristics $\mathbf{x}(\tau),\,\mathbf{p}(\tau)$
that are determined from the system \eqref{Ham_syst}--\eqref{Ham};
namely, $\mathbf{p}(\tau)=\frac{\partial s}{\partial
\mathbf{x}}(\mathbf{x}(\tau))$ and
$\frac{ds(\mathbf{x}(\tau))}{d\tau}=\langle\mathbf{p}(\tau),\frac{d
\mathbf{x}(\tau)}{d \tau}\rangle=\pm 1$.

Solutions of the system \eqref{Ham_syst}--\eqref{cond} determine
both \textit{rays}---projections of the characteristics onto plane
$(x_1,x_2)$ (the angle $\psi$ specifies the direction of the
ray)---and the \textit{wave fronts}. The wave front at time $t$ is a
curve $s(\mathbf{x})=t$. Parenthetically, the generalized phase
$S(\mathbf{x},t)$ is also determined by rays and $S(\mathbf{x},t)=0$
at front points. The choice of the initial conditions
$\mathbf{x}(t)|_{t=t_\circ}=\mathbf{x}^\circ$ for the system (we
release characteristics from the point $\mathbf{x}^\circ$) is
explained by the fact that the initial condition \eqref{WEq2} for
the wave equation \eqref{WEq1} is localized in a neighborhood of
$\mathbf{x}^\circ$. If rays do not cross and fronts are smooth
curves, then equation \eqref{HJeq} is solved by passing from
coordinates $(\tau,\psi)$ to $(x_1,x_2)$. The equating the
coefficients at $l^{-1}$ to zero yields a transport equation for the
function $f$, whose solution can also be obtained by the ray method.
However, in the vicinity of the singularity of the rays'
field---focal points of fronts---one should use other methods, in
particular, the method of the Maslov canonical operator. See the
monographs \cite[secs.~III.8--12]{MaslOp} and
\cite[secs.~I.6--8]{MaslFed} for details.

Let us give more precise definitions of the concepts mentioned
earlier. Let $\mathbf{X}(\psi,t)$ and $\mathbf{P}(\psi,t)$ be a
solution of the problem \eqref{Ham_syst}--\eqref{cond}. Then
\emph{the ray} is by definition an
$\mathbb{R}^2_{\mathbf{x}}$-subset
$\{(x_1,x_2)=(X_1(\psi_0,t),X_2(\psi_0,t)):\, t>0\}$ under the fixed
$\psi_0 \in [0,2\pi)$. At each moment of time $t$ the ends of
trajectories define smooth closed curves
$\Gamma_{t}=\{\mathbf{x}=\mathbf{X}(\psi,t),\,
\mathbf{p}=\mathbf{P}(\psi,t)\}$ in the four-dimensional phase space
$\mathbb{R}_{\mathbf{xp}}^4$.
\begin{definition}\upshape%
Curves $\Gamma_{t}$ are called \emph{the wave fronts in the phase
space}. Curves $\gamma_{t}=\{\mathbf{x}=\mathbf{X}(\psi,t), \psi \in
S^{1}\}$, which are projections of $\Gamma_{t}$ on
$\mathbb{R}_{\mathbf{x}}^2$, are called \emph{the wave fronts in the
configuration space}. The points of the fronts $\gamma_{t}$ wherein
$\pa_{\psi}\mathbf{X}=0$ will be termed the \emph{focal points}.
\end{definition}

According to \cite{DobSekTirVol, DobShafTir, DobrNaz}, at each
moment of time $t$ an asymptotic solution of the problem
\eqref{WEq1}--\eqref{WEq2} as $l \rightarrow 0$ is determined by
$\Gamma_{t}$ and has been localized in a neighborhood of
$\gamma_{t}$. Note that in contrast to $\Gamma_{t}$ the curves
$\gamma_t$ may be nonsmooth and have points of self-intersection.
Moreover, the asymptotic formulas differ in the neighborhood of
nonfocal and focal points of the front. Thus, from the point of view
of applications, a problem of visualizing the fronts is relevant.
The corresponding algorithm is described in sect.~\ref{app1}.
Formulas for the asymptotic solution of the problem
\eqref{WEq1}--\eqref{WEq2} as $l \rightarrow 0$ are given in
sect.~\ref{app2}.

 This methodology can be  further applied to the
constructing the asymptotic solutions of a  PDEs with a localized
right-hand side. To illustrate, we consider an inhomogeneous
equation
\begin{equation}\label{WEq3}
\frac{\pa^2 u}{\pa t^2}-\big\langle \nabla,
C^2(\mathbf{x})\nabla\big\rangle u=F(\mathbf{x},t),
\quad \mathbf{x}\in \mathbb{R}^2.
\end{equation}
Such equation describes the situation when a source acts over time,
while the Cauchy problem \eqref{WEq2}  for the homogeneous wave
equation \eqref{WEq1} describes the instantaneous action of the
source.  Suppose that the source is harmonic in time and localized
in space, i.~e., let
$F=\frac{1}{\varepsilon^3}e^{\frac{i}{\varepsilon}\omega t}\cdot
V(\frac{\mathbf{x}- \mathbf{x}^{\circ}}{\varepsilon})$, with
$\omega>0$. If we represent solution of equation \eqref{WEq3} in the
form $u=e^{\frac{i}{\varepsilon}\omega t}v$ we obtain the Helmholtz
equation
\begin{equation}\label{in_WEq4}
-\omega^2 v-\varepsilon^2\big\langle \nabla,
C^2(\mathbf{x})\nabla\big\rangle v=\frac{1}{\varepsilon} V
\mbig[7](\frac{\mathbf{x}-\mathbf{x}^{\circ}}{\varepsilon}\mbig[7]),
\quad \mathbf{x}\in \mathbb{R}^2.
\end{equation}
The work \cite{AnDobrNazRoul} describes an approach  to the
constructing the asymptotic solutions of inhomogeneous equations of
this type. Using the Maupertuis--Jacobi method, we transform the
original Hamiltonian
$\scr{H}(\mathbf{x},\mathbf{p})=-\omega^2+C^2(\mathbf{x})
|\mathbf{p}|^2$ for equation \eqref{in_WEq4}  (the differential
operator $\scr{H}(\mathbf{x},\hat{\mathbf{p}}),\,
\hat{\mathbf{p}}=-i\varepsilon \nabla$ defines the equation) to the
Hamiltonian \eqref{Ham}, which is considered in the present work.
The exact formulas for the fronts allow us to obtain quiet simple
expression for the asymptotic solution. An example pertaining to
this situation is discussed in sect.~\ref{app3}.

\section{Analytical solutions}\label{case}

\subsection{Underwater bank}\label{case1}

Consider first  a situation in which the bottom has the shape of
underwater bank. This means that the basin depth is defined by the
function \eqref{D1}.  Since it is symmetric, without loss of
generality we can assume that $\mathbf{x}^{\circ}=(-\xi,0)$, where
$\xi>0$. We also put the free-fall acceleration to be equal to
unity: $g=1$.

The polar symmetry guides us to pass to the polar coordinates
\begin{gather}
x_1=\varrho \cos\varphi,\quad x_2=\varrho \sin\varphi,\notag\\
p_1=u \cos\varphi-\frac{1}{\varrho}v\sin\varphi,\quad p_2=u
\sin\varphi+\frac{1}{\varrho}v\cos\varphi \label{P_E1},
\end{gather}
where $u:=p_{\varrho}$, $v:=p_{\varphi}$ are the momenta
corresponding to the variables $(\varrho,\varphi)$. In these
coordinates the Hamiltonian \eqref{Ham}  acquires the form
\begin{equation}\label{Ham_new}
\scr{H}(\varrho,\varphi;u,v)=
\sqrt{u^2+\frac{v^2}{\varrho^2}}\cdot
\sqrt{\frac{\varrho^2+b}{\varrho^2+a}}
\end{equation}
and initial conditions \eqref{cond} become
\begin{equation}\label{InD2}
u|_{t=t_\circ}=-\cos \psi,\quad v|_{t=t_\circ}=-\xi\sin
\psi,\quad \varrho|_{t=t_\circ}=\xi,\quad \varphi|_{t=t_\circ}=\pi.
\end{equation}
Since all the solutions to our models involve the elliptic integrals
and functions, we shall adopt in what follows a shortened
Weierstrass notation for them \cite{Ahiezer}:
$$
\wp(z)=\wp(z;\alpha,\beta),\quad\zeta(z)= \zeta(z;\alpha,\beta),
\quad\sigma(z)=\sigma(z;\alpha,\beta).
$$

\begin{proposition}\label{theorem1} The solution of the
Hamiltonian system \eqref{Ham_syst}--\eqref{Ham} and \eqref{Ham_new}
under the initial conditions \eqref{InD2} is given by the following
expressions
\begin{equation}\label{SOL}
\left\{
\begin{aligned}
\varrho&=\sqrt{\wp(\bop)+\delta-a},\\
 \varphi&=\pi \pm \bigg\{1+2\,b\,
\frac{\zeta(\varkappa)}{\wp'(\varkappa)}\bigg\}h\cdot(\bop-\po) \pm
\frac{b\,h}{\wp'(\varkappa)}
\ln\frac{\sigma(\bop-\varkappa)\,\sigma(\po+\varkappa)}
{\sigma(\bop+\varkappa)\,\sigma(\po-\varkappa)}\\
u&=-\frac12
\sqrt{\frac{\xi^2+b}{\xi^2+a}}\,
\frac{\wp'(\bop)}{\wp(\bop)+\delta-a+b}\,
\frac{1}{\sqrt{\wp(\bop)+\delta-a}},\quad
v=-\xi\sin\psi
\end{aligned}
\right.,
\end{equation}
The functions $\bop$, $\po$, and $\varkappa$ are solutions of the
transcendental equations
\begin{equation}\label{p0pkappa}
\wp(\po)= \xi^2-\delta+a\,, \quad \wp(\varkappa)=\frac13(a+b-\frak h)\,,\quad
t-t_\circ=\delta\cdot (\bop-\po)-\zeta(\bop)+\zeta(\po)\,.
\end{equation}
The expressions for  $\delta$, $\alpha$ and $\beta$ in terms of
parameters of the problem are as follows
\begin{equation}\label{alpha}
\left\{
\begin{aligned}
\delta&=\frac13\,({\frak h}+2\,a-b), \quad
\alpha=\frac43\,\big({\frak h}^2-
2\,(a-2\,b)\,{\frak h}+a^2-a\,b+b^2\big)\\[1ex]
\beta&=\frac{4}{27}\,\big({\frak h}-a+2\,b\big)
\big(2\,{\frak h}^2-4\,(a-2\,b){\frak h}+(2\,a-b)(a+b)\big)
\end{aligned}
\right.,
\end{equation}
where
\begin{equation}\label{h}
h^2\FED{\frak h}=\xi^2\,\frac{\xi^2+a}{\xi^2+b}\sin^2\psi\,.
\end{equation}
\end{proposition}

\begin{remark}
\upshape All the quantities $\po$, $\bop$, and $\varkappa$ may be
complex. They lie on the edges of the parallelogram
$(0,\omega,\omega+\omega',\omega'),$ where $\omega$ is the pure real
period of the Weierstrass  function and $\omega'$ is pure imaginary.
Besides, $\bop$ lies on the same edge of the parallelogram as $\po$.
It should be noted that equations \eqref{p0pkappa} have infinitely
many roots. Because of this, we assume that the roots from the first
positive branch are taken. An algorithm which illustrates the
formulas and this remark is discussed in sect.~\ref{app1}.
\end{remark}
\begin{proof}
 Since transformation \eqref{P_E1} is canonical, the polar
representation of the system under consideration has the same
gradient form as \eqref{Ham_syst}. Solution of
\eqref{Ham_syst}--\eqref{cond} thus boils down to integrating the
equations
\begin{alignat}{3}
\dot u&=-\scr H_\varrho\,,&\qquad&\dot\varrho=\frac{u\,\varrho}
{\sqrt{u^2\varrho^2+v^2}}\,
\sqrt{\frac{\varrho^2+b}{\varrho^2+a}}\qquad(=\scr H_u), \notag \\
\label{Ham_sys_new} \dot v&=0\,,&\qquad&\dot
\varphi=\frac{(v/\varrho)}{\sqrt{u^2\varrho^2+v^2}}\,
\sqrt{\frac{\varrho^2+b}{\varrho^2+a}}\qquad(=\scr H_v)
\end{alignat}
with the Hamiltonian function \eqref{Ham_new} and initial conditions
\eqref{InD2}. We however do not use the canonicity of this
transformation, because the mere form \eqref{Ham_sys_new} suggests
the scheme of obtaining the separable dynamical equations. For the
same reason, nor do we resort to the standard separability theory
(see, e.~g. \cite{blaszak}),

Indeed, the following laws of conservation are obvious:
\begin{equation}\label{ConsL1}
|p|c(\varrho)=c_\circ,\quad
c_\circ=c(|\xi|)\equiv\sqrt{\frac{b+\xi^2}{a+\xi^2}},\quad
v=-\xi \sin \psi\,,
\end{equation}
and we represent them through the free complex constants $\gamma$
and $h$:
\begin{equation}\label{int}
v=\gamma\quad(=\const)\,,\qquad \scr{H}^2(\varrho,\varphi; u,
\gamma)=\mbig[7](u^2+\frac{\gamma^2}{\varrho^2}\mbig[7])
\frac{\varrho^2+b}{\varrho^2+a}=\frac{\gamma^2}{h^2} \quad( = \const)\,.
\end{equation}
These constants are related with the initial conditions, and it is
not difficult to see that these relationships are  equivalent to the
following ones:
$$
\gamma=-\xi \sin \psi,\quad h^2=\xi^2 \frac{a+\xi^2}{b+\xi^2}\sin^2\psi.
$$

Let us ascertain the dynamics $\varrho=R(t)$. By rewriting the
integral \eqref{int} in the form
$\scr H(\varrho,\varphi;u,v)=\gamma/h$, we obtain identity
$\sqrt{\frac{\varrho^2
+b^2}{\varrho^2+a^2}}=\frac{\gamma}{h\sqrt{u^2+\gamma^2/\varrho}}$,
whence one gets
\begin{equation}\label{xu}
\dot\varrho=\frac{(\gamma/h)\cdot u}{u^2+\gamma^2/\varrho^2}.
\end{equation}
Expressing $u$ via variable $\varrho$ and using the integral
$\scr H(\varrho,\varphi;u,v)=\gamma/h$, we obtain the autonomous
dynamics for the function $\varrho$:
\begin{equation*}
\varrho\,\dot\varrho=
\frac{\sqrt{(\varrho^2+b)\mbig[1][\varrho^2(\varrho^2+a)-
(\varrho^2+b)h^2\mbig[1]]}}{(\varrho^2+a)}\,.
\end{equation*}
The change $\varrho^2+a\FED z$ suggests itself, and we derive
\begin{equation*}
\frac12\,\dot z=\frac{1}{z}\,\sqrt{(z-a+b)\mbig[1][z\,(z-a)-(
z-a+b)h^2\mbig[1]]}\,.
\end{equation*}
This dynamics is readily transformed into the integral form
\begin{equation}\label{base}
\int\limits^z_{\infty}{z}\,\frac{dz}{w}=t\,,\qquad
\begin{aligned}
w^2\DEF{}& 4\,\mbig[1][z\,(z-a)-(z-a+b)h^2\mbig[1]](z-a+b)=\\
={}&4\,(z-A)(z-B)(z-C)
\end{aligned}\,,
\end{equation}
where constants $(A,B,C)$ can be expressed through the parameters of
the problem $(a,b,\xi)$ and constants $(\gamma,h)$. As was mentioned
in Introduction, we arrive at the problem of inversion of a
meromorphic elliptic integral.

Let us reduce \eqref{base} to the canonical Weierstrass form. For
this we make a shift
\begin{equation*}
\begin{aligned}
z\goto s&=z-\delta\,,\quad \delta\DEF\frac13\,(A+B+C)
\end{aligned}
\end{equation*}
in order to obtain the standard:
$4\,s^3-\alpha\,s-\beta=4\,(z-A)\,(z-B)\,(z-C)$. Hence it follows that
the new constants $(\alpha,\beta)$ are given by the expressions
\eqref{alpha}. Then equation \eqref{base} takes the form
\begin{equation*}
\int\limits^{\varrho^2+a-\delta}_\infty
\![8]\frac{(s+\delta)\,ds}{\sqrt{4\,s^3-\alpha\,s-\beta}}=t
\qquad\goto\qquad \varrho=R(t)\,
\end{equation*}
and its solution can be represented  in terms of $\wp$-function as
\begin{equation}\label{R(t)}
\varrho=R(t)=\sqrt{\wp\big({\frak p}(t);\alpha,\beta\big)+\delta-a}.
\end{equation}
Here, the function ${\frak p}(t)$ is determined---in the full
elliptic notation---from equation
\begin{equation}\label{trans}
\delta\cdot \bop-\zeta(\bop;\alpha,\beta)=t,
\end{equation}
and $\zeta(\bop;\alpha,\beta)$ is the associated Weierstrass
zeta-function. It is related to the $\wp$-function by definition
$\zeta'(\bop;\alpha,\beta)=-\wp(\bop;\alpha,\beta)$. Let us adopt
the solution $\varrho=R(t)$, \ie, formulas
\eqref{R(t)}--\eqref{trans} for the case of  an arbitrary initial
condition $R(t_\circ)=\varrho_\circ$. It is more  convenient to
represent it in a $\bop$-parametric form:
\begin{gather}\label{xt}
\notag
\begin{aligned}
\varrho=\sqrt{\wp(\bop)+\delta-a}\,,\quad
t-t_\circ=\delta\cdot (\bop-\po)-\zeta(\bop)+\zeta(\po),
\end{aligned}\\[1ex]\notag
\po\,:\quad \sqrt{\wp(\po)+\delta-a}=\varrho_\circ\qquad\text{(initial
condition)}\,.
\end{gather}
The solution $u=U(t)$ can be derived from \eqref{int}--\eqref{xu}
and from the relation
$$
\dot{\bop}(t)=\frac{1}{\wp(\bop(t))+\delta}.
$$
Substituting the expression for $\varrho=R(t)$ we get the 3-rd
formula in \eqref{SOL}.

Consider now the solution $\varphi = \Phi(t)$. With use of the last
equation of system \eqref{Ham_sys_new} and expression \eqref{R(t)},
we obtain
\begin{equation*}
\dot \varphi=\frac{h}{\varrho^2}\,\frac{\varrho^2+b}{\varrho^2+a}=
h\mbig[7](1+ \frac{b}{\wp(\bop)+\delta-a}
\mbig[7])\frac{1}{\wp(\bop)+\delta}\,.
\end{equation*}
Note that we have used the integral \eqref{int} as a radical.
Therefore the sign should be taken into consideration, and we
understand $\pm h$ by the notation $h$. Thus,
\begin{equation*}
\varphi=\varphi_\circ^{}+ h\,(\bop-\po)+b\,h\cdot
\int\limits^t_{t_\circ}\frac{1}{\wp\mbig[1](\frak p(t)\mbig[1])+
\delta-a}\cdot\frac{d t}{\wp\mbig[1](\frak p(t)\mbig[1])+ \delta}.
\end{equation*}
Making the change of variable by equality
$dt=(\wp(\bop)+\delta)\,d\bop$, one obtains
\begin{equation*}
\varphi=\varphi_\circ^{}+h\,(\bop-\po)+b\,h\cdot \int\limits^{\frak
p(t)}_{\po}\frac{d\bop}{\wp(\bop)- \lambda},\qquad \lambda\DEF
a-\delta\,.
\end{equation*}
This is nothing but the \emph{logarithmic} elliptic integral, which
can be easily evaluated \cite{Ahiezer}. We then arrive at formula
\begin{equation*}
\varphi=\varphi_\circ^{}+h\,(\bop-\po)+
\left.\frac{b\,h}{\wp'(\varkappa)}
\bigg\{\ln\frac{\sigma^2(\bop-\varkappa)}{\sigma^2(\bop)}-
\ln\big(\wp(\bop)-\lambda\big)+2\,\zeta(\varkappa)\,\bop\bigg\}
\right|_{\po}^{\bop}\,,
\end{equation*}
where $\varkappa$ is defined by the equality
$\lambda=\wp(\varkappa)$. This expression can be rewritten in terms
of Weierstrass $\sigma$-functions alone:
\begin{equation}\label{angle}
\notag
\begin{aligned}
\varphi={}&\varphi_\circ^{} \pm \bigg\{1+2\,b\,
\frac{\zeta(\varkappa)}{\wp'(\varkappa)}\bigg\}h\cdot(\bop-\po) \pm
\frac{b\,h}{\wp'(\varkappa)}
\ln\frac{\sigma(\bop-\varkappa)\,\sigma(\po+\varkappa)}
{\sigma(\bop+\varkappa)\,\sigma(\po-\varkappa)}\,,\\[1ex]
\varkappa\DEF{}&\wp^{-1}\big(\tfrac13(a+b-\frak h)\big)\,.
\end{aligned}
\end{equation}
Invoking the initial condition \eqref{InD2}, we get finally the
solution \eqref{SOL}.
\end{proof}

\subsection{Underwater ridge}\label{case2}

Consider now the underwater ridge. This means that the basin depth
is defined be the function \eqref{D2}. As earlier we set the
free-fall acceleration to unity. Then the Hamiltonian \eqref{Ham}
has the form
\begin{gather}\label{Ham2}
\scr H(x_1,x_2;p_1,p_2)=\sqrt{p_1^2+p_2^2}\cdot
\sqrt{\frac{x_1^2+b}{x_2^2+a}}\,,\\
(p_1,p_2)\DEF (p_{x_1},p_{x_2}).\notag
\end{gather}
As in the previous section, by symmetry and without loss of
generality, we may set $\mathbf{x}^{\circ}=(-\xi,0),$ $\xi>0$. The
initial conditions \eqref{cond} are thus as follows
\begin{equation}\label{incond_2}
x_1|^{}_{t=t_\circ}=-\xi,\qquad x_2^{}|_{t=t_\circ}=0,\qquad
p_1|^{}_{t=t_\circ}=\cos \psi, \qquad
p_2|^{}_{t=t_\circ}=\sin\psi.
\end{equation}

\begin{proposition}
The solution of the Hamiltonian system \eqref{Ham_syst}--\eqref{Ham}
and \eqref{Ham2} under the initial conditions \eqref{incond_2} is
given by the following expressions
\begin{equation}\label{res_case2}
\left\{
\begin{aligned}
x_1&=-\sqrt{\wp(\bop;\alpha,\beta)+\delta-a},\quad
x_2=\sqrt{1-\frak h^{-1}}\cdot t+ \sqrt{\frak
h-1}\,(b-a)\cdot(\bop-\po)\\
p_1&=\frac{1}{2}\,\frac{\gamma}{\sqrt{\frak h-1}}
\frac{\wp'\big({\frak
p}(t)\big)} {\wp\big({\frak p}(t)\big)+\delta-a+b}
\sqrt[-2]{\wp\big({\frak p}(t)\big)+\delta-a},\quad
p_2=\gamma
\end{aligned}
\right.\,.
\end{equation}
The functions $\bop$ and $\po$ are solutions of transcendental
equations
\begin{equation}\label{new2'}
\frac{1}{h}\,(t-t_{\circ})=\delta\cdot
(\bop-\po)-\zeta(\bop;\alpha,\beta)+ \zeta(\po;\alpha,\beta)\,,\quad
\po=\wp^{-1}(\xi^2-\delta+a;\alpha,\beta)\,.
\end{equation}
The expressions for $\delta$, $\alpha$ and $\beta$ in terms of
parameters of the problem are as follows
\begin{equation}\label{alpha'}
\left\{
\begin{aligned}
\delta&=\frac13\,\big((b-a)\, \frak h+3\,a-2\,b\big), \quad
\alpha=\frac{4}{3}\,\big((b-a)^2\frak h^2-
(b-a)\,b\,\frak h+b^2\big)\\[1ex]
\beta&=\frac {4}{27}\,\big((b-a)\,\frak h+b\big)
\big((b-a)\,\frak h-2\,b\big)\big(2\,(b-a)\,\frak h-b\big)
\end{aligned}\right.\,,
\end{equation}
where
\begin{equation*}
\frak{h}=\frac{\xi^2+b}{(\xi^2+b)-\sin^2\psi(\xi^2+a)}, \quad
\gamma=\sin \psi.
\end{equation*}
\end{proposition}

\begin{proof}
The Hamiltonian system with the Hamiltonian \eqref{Ham2} has the
form
\begin{alignat}{3}
\dot{p}_1&=-\scr H_{x_1}\,,&\qquad&\dot{x}_1=
\frac{p_1}{\sqrt{p_1^2+p_2^2}}\,
\sqrt{\frac{x_1^2+b}{x_1^2+a}}\qquad(=\scr H_{p_1}),\notag \\
\label{Ham_syst_case2}
\dot{p}_2&=0\,,&\qquad&\dot{x}_2=\frac{p_2}{\sqrt{p_1^2+{p_2}^2}}\,
\sqrt{\frac{x_1^2+b}{x_1^2+a}}\qquad(=\scr H_{p_2}).
\end{alignat}
Using laws of conservation, we obtain
\begin{equation}\label{int'}
p_2=\gamma\qquad(=\const)\,,\qquad
(p_1^2+\gamma^2)\,\frac{x_1^2+b}{x_1^2+a}=\frac{\gamma^2\,h^2}{h^2-1}
\qquad( = \const)\,.
\end{equation}
Here, $\gamma$ and $h$ are the complex constants related to initial
conditions of the problem through the formulas
\begin{equation}\label{InD5}
\gamma=\sin \psi, \qquad
\frac{\xi^2+b}{\xi^2+a}=\frac{\gamma^2\,\frak h}{\frak h-1}\qquad
(\frak h\DEF h^2)\,,
\end{equation}
and $(a,b,\xi)$ are parameters of the problem.

Let us ascertain the dynamics $x_1=X_1(t)$. Using the integral
$\scr H(x_1,x_2;p_1,p_2)=\frac{\gamma h}{\sqrt{\frak h -1}}$, we
obtain
\begin{equation}\label{tmp}
\dot{x}_1=\frac{p_1}{p_1^2+\gamma^2}\cdot
\frac{\gamma\,h}{\sqrt{\frak h-1}}\,.
\end{equation}
Acting in the same way as in the previous case, we derive the
dynamics for the variable $x_1$
$$
x_1\,\dot{x}_1=\frac{\sqrt{x_1^2\big(x_1^2+b\big)
\big(x_1^2-(\frak{h}-1)(b-a)\big)}}{h(x_1^2+a)}.
$$
The change $z\DEF x_1^2+a$ gives rise again to a meromorphic
integral:
\begin{equation*}\label{base'}
\int\limits_{t_\circ}^z z\,\frac{dz}{w}=\frac1h\cdot (t-t_\circ)\,,\qquad
\begin{aligned}
w^2\DEF{}& 4\,(z-a)(z-a+b)\big(z-(\frak h-1)\,(b-a)\big)=\\
={}&4\,(z-a)(z-a+b)(z-C).
\end{aligned}
\end{equation*}
As previously, we reduce this expression to the canonical form by a
shift $z\goto s=z-\delta$ in order to obtain the Weierstrassian
standard: $4\,s^3-\alpha\,s-\beta=4\,(z-a)\,(z-a+b)\,(z-C)$. Hence, the
constant $\delta$ and invariants of elliptic functions $\alpha$ and
$\beta$ are calculated through $a, b, \xi,\psi$; these are
expressions \eqref{alpha'}.

The solution has the same structure as in the previous case:
\begin{gather}\label{xt'}
\notag
\begin{aligned}
x_1=\pm\sqrt{\wp\big(\frak p(t);\alpha,\beta\big)-(a-\delta)},\\[1ex]
\delta\cdot
(\bop-\po)-\zeta(\bop;\alpha,\beta)+\zeta(\po;\alpha,\beta)&=
\frac1h\cdot (t-t_\circ),
\end{aligned}\\[1ex]\notag
\po\,:\quad \sqrt{\wp(\po)+\delta-a}=\xi \qquad\text{(initial
condition)}\,.
\end{gather}
Using \eqref{int'}, \eqref{tmp} and the relation
$$
\dot {\frak p}(t)=\frac1h\,\frac{1}
{\wp\mbig[1]({\frak p}(t)\mbig[1])+\delta}\,,
$$
we do substitute $x_1=X_1(t)$ found above and obtain solution
$p_1=P_1(t)$ displayed in \eqref{res_case2}.

Let us determine the dynamics $x_2=X_2(t)$. From the last equation
of the system \eqref{Ham_syst_case2} one obtains
\begin{equation*}
\dot{x}_2=\sqrt{1-\frak h^{-1}}\,\frac{x_1^2+b}{x_1^2+a}=
\sqrt{1-\frak h^{-1}}\mbig[7](1+\frac{b-a}{\wp(\bop)+\delta}
\mbig[7])\,.
\end{equation*}
Hence,
$$
x_2=x_2^{\circ}+\sqrt{1-\frak h^{-1}}\cdot(t-t_\circ)+\sqrt{1-\frak
h^{-1}}\,(b-a)\cdot
\int\limits^t_{t_\circ}\frac{dt}{\wp\mbig[1]({\frak p}(t)\mbig[1])+
\delta}\,.
$$
Rewriting the last integral in terms of elliptic functions, we
arrive at \eqref{res_case2}. No use is required of the
$\sigma$-functions in this case.
\end{proof}

\section{Applications to asymptotic theory}\label{app}

In this section we discuss some applications of obtained analytical
solutions to the asymptotic theory. We rely on the results by
Dobrokhotov, Nazaikinskii, Shafarevich, and by others
\cite{AnDobrNazRoul,DobrNaz,DobrNazTir,DobSekTirVol,DobShafTir} in
their study of the linear shallow-water waves, which are generated
by a localized source and propagate in the basin with a variable
depth. The method used is based on a modification of the Maslov
canonical operator \cite[secs.~I.6--8]{MaslFed}. One of the
advantages of this method is that it gives the convenient asymptotic
formulas in the neighborhood of focal points and caustics. Thus, the
analytical formulas for solutions of the Hamiltonian systems under
consideration allow us, among other things, to seek focal points. We
also note that the turning points at the wave fronts are focal.
Thereby visualization of the fronts can be considered as one of the
applications.

\subsection{Visualization of fronts. The algorithm}\label{app1}

\begin{figure}[h!]
\begin{minipage}[h]{0.39\linewidth}
\includegraphics[width=1\linewidth,height=4cm]{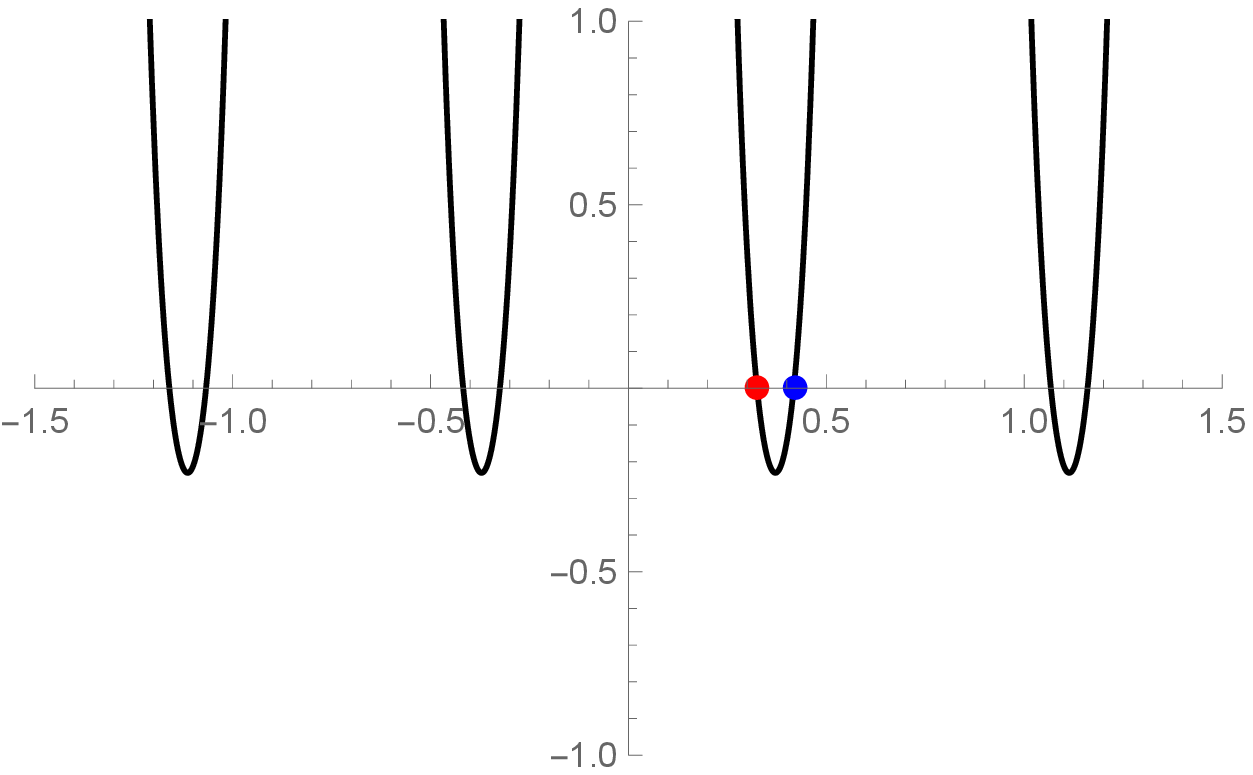}
\caption{Graph $y=\wp(x)-\xi^2+\delta-a$ under
$\psi=\frac{9\pi}{10}$} \label{figp0}
\end{minipage}
\hfill
\begin{minipage}[h]{0.59\linewidth}
\begin{minipage}[h]{0.64\linewidth}
\center{\includegraphics[width=0.55\linewidth]{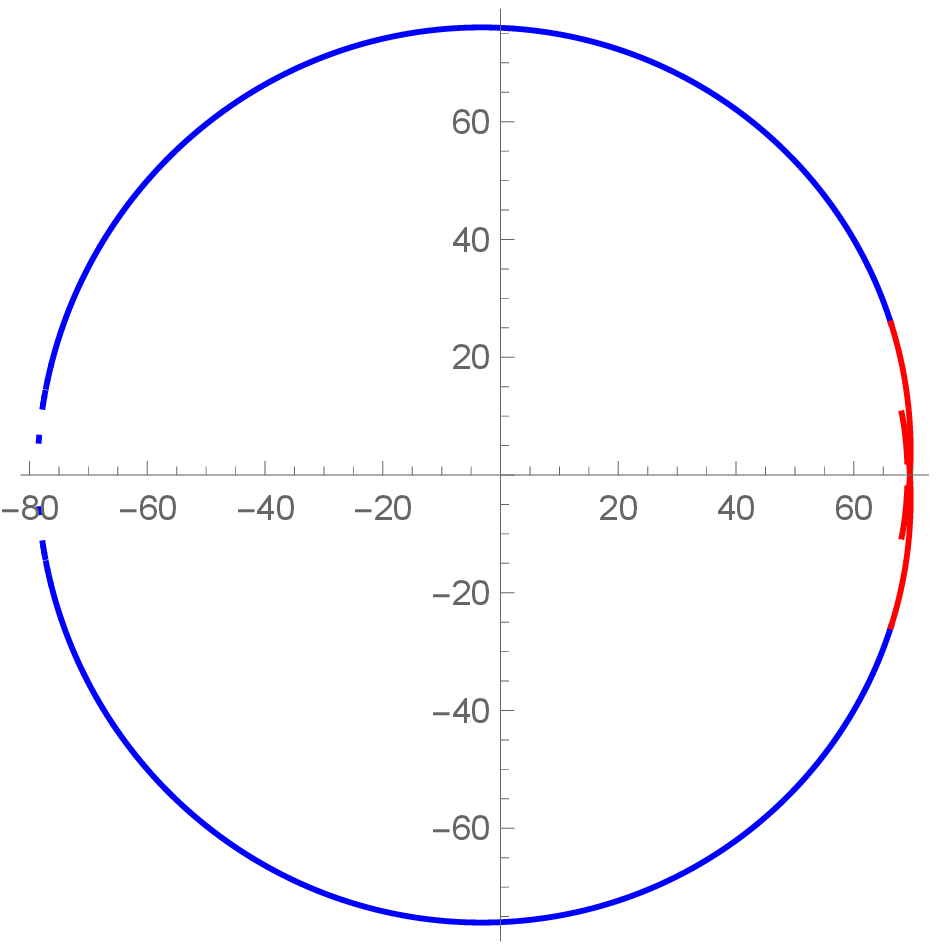}}
\end{minipage}
\hfill
\begin{minipage}[h]{0.34\linewidth}
\center{\includegraphics[width=0.9\linewidth,height=4cm]{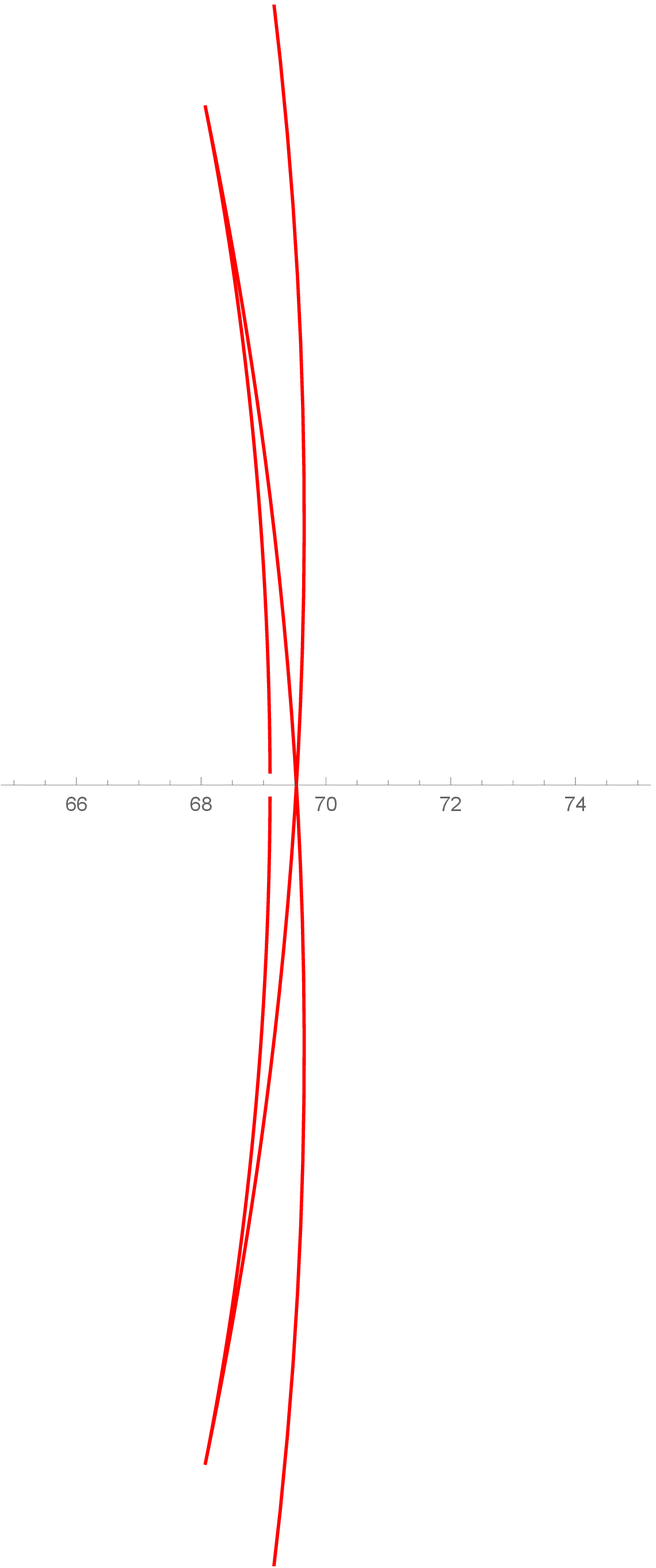}}
\end{minipage}
\caption{Graph of front under $t=100$}
\label{figt100}
\end{minipage}
\end{figure}
Here we describe an algorithm for constructing the fronts
$\gamma_t$; it has been implemented with use of the software package
\textsc{Wolfram Mathematica}  and intensively exploits the elliptic
transcendents. Except for visualizating the fronts and focal points,
the algorithm can be applied to the asymptotic theory and allows us
to determine the structure of Lagrangian manifolds
\cite[sect.~2.1]{Mishenko} and to construct the asymptotic solutions
of other problems. In particular, the problem
\eqref{WEq1}--\eqref{WEq2} is solved with use of the ray method or
the Maslov canonical operator theory, which is its generalization.
We will briefly describe the idea of Maslov's theory further bellow.
Moreover, as mentioned earlier, the asymptotic solution of the
problem \eqref{WEq1}--\eqref{WEq2} is localized in the neighborhood
of the front \cite{DobSekTirVol, DobShafTir, DobrNaz}.
\bigskip

\textbf{Algorithm:}

1. Fix first the parameters of the problem and, by way of
illustration, consider the case of underwater bank
(sect.~\ref{case1}). Constants $a$, $b$ determine the bottom shape
and constant $\xi$ define the source position. We take the following
values: $a=100$, $b=1$, and $\xi=0.5.$

2. For each fixed  $\psi \in [0,\pi]$, there exist infinitely many
roots of the equation
\begin{equation}\label{alg_p2}
\wp(p_0)=\xi^2-\delta+a.
\end{equation}
As values of $p_{01}(\psi)$ and $p_{02}(\psi)$ we choose the first
and second positive roots respectively (see Fig.~\ref{figp0}). In
this figure, the red point corresponds to $p_{01}$ and the blue
point corresponds to $p_{02}$. The quantity $\varkappa(\psi)$ is
determined by the formula \eqref{p0pkappa}, where we take again the
first positive root.

\begin{remark}\upshape
As mentioned above, the roots of the equation \eqref{alg_p2} need
not to be real. We have to seek the complex roots in one of the
forms $\ri\,p_0$, $p_0+\omega'$, or $\ri\,p_0+\omega$, where $\omega$
is the pure real period of the Weierstrassian $\wp$, and $\omega'$
is its pure imaginary period; the $p_0$ is real. Hence we obtain
some equation for $p_0$, which has infinitely many solutions. As
$p_{01}$ and $p_{02}$ we must take, respectively, the nearest first
and second positive roots (roots from the first positive branch).
\end{remark}

3. Fix further the point of time $t$. Let $p_{i}(\psi,t)$ $(i=1,2)$
be the first positive root of the equation $$t=\delta\cdot
(p-p_{0i}(\psi))-\zeta(p)+\zeta(p_{0i}(\psi)) \quad (i=1,2).$$ If
$p_{0i}$ is not real, then we should seek $p_{i}$ in the same manner
as $p_{0i}$; \ie, on the same edge of the parallelogram.

4. For each fixed $t$ we plot the curves in polar coordinates
$\varrho,\varphi$ by formula \eqref{SOL}. In the case $\psi \in
\left[0,\frac{\pi}{2}\right]$ we take $p_{01}$ and $p_1$ as
parameters, and in the case  $\psi \in
\left[\frac{\pi}{2},\pi\right]$ we take $p_{02}$ and $p_2$.

5. By symmetry, we reflect the obtained graph about the horizontal
axis $Ox$. Figure~\ref{figt100} exhibits a front at $t=100$; the
blue curve in this figure corresponds to $\psi \in
\left[0,\frac{\pi}{2}\right]$ and the red curves correspond to $\psi
\in \left[\frac{\pi}{2},\pi\right]$.

\begin{remark}\upshape
 Once more to emphasize, the inversion of integrals/functions
we have met in the algorithm is a necessary attribute for
representing the solutions. As is well known, the explicit
$t$-dynamics calls for inversion of integrals \cite{belokolos, br1}.
It is this point that creates the $t$-dependence per~se in the
Liouvillian integrability \cite{br1, blaszak}; even the trivial
harmonic oscillator is not an exception. In turn, the problem of
wave fronts in and of itself requires the elimination of initial
data from this $t$-dependence and, thereby, also deals with further
transcendental act. With that, this second transcendence is not a
trivial inversion of the first one; they have \emph{different}
nature. We thus arrive at the `double complication', which is
however an inevitable property of any analytic theory, not a
specific feature of the elliptic case. To get an analytic writing
for both the transcendences, as have been pointed out in
Introduction, is a nontrivial task. Our elliptic formulas correspond
to a situation when all the inversion problems have an explicit
function-shape, and it is these formulas that allow us to advance in
the problem of asymptotics; see below. It is in this sense that we
mean the problem as exactly solvable. In other words, just as the
class of functions in use is supplemented with
$\{\Theta,\wp,\zeta,\ldots\}$-objects \cite[sect.~6]{br1} in the
classical solvability \cite{belokolos}, so also we should supplement
the theory of fronts with new operations---solutions of our
transcendental equations. No matter what that procedure may be,
numerical or non.
\end{remark}

\subsection{Cauchy problem with localized initial conditions}\label{app2}

As just mentioned, the exact analytical formulas for the solution of
the problem \eqref{Ham_syst}--\eqref{cond} allow one to obtain an
expression for the asymptotic solution of the problem
\eqref{WEq1}--\eqref{WEq2} as $l \rightarrow 0$. Asymptotic formulas
for such problems were discussed, for example, in
\cite{DobrNaz,DobSekTirVol,DobShafTir}. These works are based on the
theory of the Maslov canonical operator. This operator describes
asymptotic solutions to a wide class of problems for differential
equations. In turn, the method of constructing the Maslov operator
is a generalization of the ray method and of the
Wentzel--Kramers--Brillouin--Jeffreys (WKBJ) approximation
\cite{Heading}.

By using the WKBJ approximation or an approximation similar to that
given in the Introduction, one can define the Hamilton function
corresponding to the differential operator. A surface formed by the
phase trajectories of the corresponding dynamical system does then
determine the Lagrangian manifold. If this manifold is projected
onto the physical space  ambiguously, \ie., when the focal points
and caustics do appear, the standard methods are not applicable.
Meantime, the canonical operator allows one to avoid this problem by
a rotating of the coordinate system. Thus, the construction of the
canonical operator is defined by the structure of the Lagrangian
manifold. Since the canonical-operator method exploits the rather
complicated mathematical technique even at the level of strict
definitions \cite[Part I]{MaslFed}, we do not expound it here and
use only some consequences stemmed from this method.

In this subsection, we present a result for one of the bottom types
considered above---the underwater ridge. That is, let function $D$
be defined by the expression \eqref{D2}; again and without loss of
generality, we assume that $g=1$ and $\mathbf x^{\circ}=(-\xi,0)$
with $\xi>0$. Then, using the result of the work
\cite{DobSekTirVol}, we arrive at the following statement.

\begin{proposition}
For $t>\sigma>0$ in the neighborhood of nonfocal point of the front
$\gamma_t$ the asymptotic solution to the problem
\eqref{WEq1}--\eqref{WEq2} has the form
\begin{equation}\label{asymp}
\begin{aligned}
w(\mathbf{x},t) &=
\left.\frac{l^{\sss\frac12}}{|\mathbf{X}'_{\psi}(\psi,t)|^\frac12}
\left(\frac{b(x_1^2+a)}{a(x_1^2+b)}\right)^{\frac14}
\mathrm{Re}\left[e^{-i\pi m(\psi,t)/2}
F\left(\frac{S(\mathbf{x},t)}{l},\psi\right)\right]
\right|_{\psi=\psi(\mathbf{x},t)}\\
&\quad+O\mbig[7](\frac{l}{|\mathbf{X}- \mathbf{x}^{\circ}|}\mbig[7]).
\end{aligned}
\end{equation}
Here, outside this region $w(\mathbf{x},t)=O
\big(\frac{l}{|\mathbf{X}-\mathbf{x}^\circ|}\big)$ and function $F$
is defined as follows
$$
F(z,\psi) = \frac{e^{-i\pi/4}}{\sqrt{2\pi}}\int\limits_{0}^{\infty}
\tilde{V}(s\mathbf{n}(\psi))e^{izs}ds\,,\quad
\tilde{V}(k)=\frac{1}{\sqrt{2\pi}}
\int\mkern-12mu\int_{\mathbb{R}^2}V(s) e^{-i\langle s,k\rangle}ds.
$$
The phase $S$ has the form
$$
S(\mathbf{x},t)=\big\langle\textbf{P}(\psi(\mathbf{x},t),t),
\mathbf{x}-\mathbf{X} (\psi(\mathbf{x},t),t)\big\rangle=
\left(\frac{b(X_1^2(\psi(\mathbf{x},t),t)+a)}
{a(X_1^2(\psi(\mathbf{x},t),t)+b)}\right)^{\frac12}L(\mathbf{x},t),
$$
where $L(\mathbf{x},t)$ is the distance between point $\mathbf{x}$
and the wave front $\gamma_t$. Moreover, we  put $L>0$ for the
external subset of the front and $L<0$ for the internal subset, the
$\psi(\mathbf{x},t)$  is defined by the condition
$\langle\mathbf{x}-\mathbf{X}(\psi,t),
\mathbf{X}'_{\psi}(\mathbf{x},t) \rangle =0$, and $m(\psi,t)$ is the
Maslov index coinciding with the Morse index of the trajectory
$(\mathbf{X}(\psi,\tau),\mathbf{P}(\psi,\tau)),\,\tau \in (0,t)$
under the notation
$\mathbf{X}(\psi,t)\DEF(x_1(\psi,t),x_2(\psi,t))$,
$\mathbf{P}(\psi,t)\DEF(p_1(\psi,t), p_2(\psi,t))$.  All the
functions $x_1,\,x_2,\,p_1,\,p_2$ are defined by the elliptic
formulas~\eqref{res_case2}.
\end{proposition}

The expression for an asymptotic solution in a neighborhood of a
focal point is more cumbersome; we do not display it here. The case
of arbitrary $D$-functions is elaborated in the work
\cite{DobSekTirVol} and can be  adopted to the cases we consider.

 As was mentioned earlier, asymptotics similar to
\eqref{asymp} and a similar formula in the vicinity of focal points
can be obtained for the solution of the wave equation \eqref{WEq1}
with any bottom-depth function $D$ if it corresponds to an
integrable Hamiltonian. Such asymptotics is determined by the
solutions to this dynamical system and by the profile function $F$
according to the choice of the source function $V$. Of course, the
generic Hamiltonian system can be integrated by numerical methods.
But we stress that the exact solvability of Hamiltonians \eqref{Ham}
with the $D$-functions \eqref{D1}--\eqref{D2} entails the
\textit{completely analytical} formula for the asymptotics. These $D$'s
are not trivial and describe bottom inhomogeneities, the presence of
which gives birth to the focal points and caustics.

\subsection{Helmholtz equation with a localized right-hand
side}\label{app3}

Yet another application of the Hamiltonian system above comes from
the search for asymptotic solutions to the wave equation
\eqref{WEq3} with a localized right-hand side. Such problems were
discussed, for example, in works \cite{DobrNaz, DobrNazTir}. As
mentioned in Introduction, we assume that the source is harmonic in
time. Then the wave equation can be reduced to the Helmholtz
equation \eqref{in_WEq4}. We do also assume that at infinity the
sought-for solution satisfies conditions of the asymptotic limiting
absorption principle \cite[p.~407]{AnDobrNazRoul}, which is an
asymptotic analogue to the  standard principle of limiting
absorption \cite{Eidus}. As applied to our method this means that
the trajectories of the corresponding Hamiltonian system do not lie
in the compact space.

We present a result for the case when the bottom has the form of an
underwater bank:
\begin{equation}\label{c}
C^2(|\mathbf{x}|) = \frac{|\mathbf{x}|^2+b}{|\mathbf{x}|^2+a}.
\end{equation}
As elsewhere, without loss of generality, we put
$\mathbf{x}^{\circ}=(-\xi,0)$, $\xi >0$. Suppose that the right-hand
side has the form of a non-symmetric `Gauss bell':
\begin{equation}\label{V}
V(y)=\exp\left[-\frac{1}{2}\left(\frac{y_1^2}{c^2}+\frac{y_2^2}{d^2}
\right)\right],
\end{equation}
where $c$ and $d$ are constants.

To describe the asymptotic solution of equation \eqref{in_WEq4} with
the velocity \eqref{c} and right-hand side \eqref{V} as $\varepsilon
\rightarrow 0$, we use the Maupertuis--Jacobi principle and a result
of the work \cite{AnDobrNazRoul}. This article describes an approach
 for obtaining asymptotics of the stationary problems with
localized right-hand sides by use of the Maslov canonical operator
on a pair of Lagrangian manifolds.

The equation \eqref{in_WEq4} can be rewritten as follows
$$
\scr{H}(\mathbf{x},\hat{\mathbf{p}})v(\mathbf{x})=
\frac{1}{\varepsilon}V\mbig[7](\frac{\mathbf{x}-\mathbf{x}^{\circ}}
{\varepsilon}\mbig[7]),\quad \scr{H}(\mathbf{x},\hat{\mathbf{p}})=
-\omega^2+\big\langle\hat{\mathbf{p}}, C^2(|\mathbf{x}|)
\hat{\mathbf{p}}\big\rangle,\quad
\hat{\mathbf{p}}=\mbig[7]({-}i\varepsilon \frac{\pa }{\pa x_1},
-i\varepsilon \frac{\pa }{\pa x_2}\mbig[7]).
$$
Thus, the symbol of the operator has the form
$\scr{H}(\mathbf{x},\mathbf{p})=-\omega^2+C^2(|\mathbf{x}|)
(p_1^2+p_2^2).$ Notice that instead of
$\scr{H}(\mathbf{x},\mathbf{p})$ we  may consider  the Hamiltonian
$$
H(\mathbf{x},\mathbf{p})=\sqrt{\frac{|\mathbf{x}|^2+b}
{|\mathbf{x}|^2+a}}\cdot|\mathbf{p}|\,,
$$
which was studied in  sect.~\ref{case1}. Since the level surfaces
$\scr{H}(\mathbf{x},\mathbf{p})=0$ and
$H(\mathbf{x},\mathbf{p})=\omega$ coincide, the phase trajectories
of the vector fields $V_{\scr{H}}$ and $V_{H}$, in accord with the
Maupertuis--Jacobi principle, do also coincide. Therefore
$$
\Lambda_{+}=\Ccup_{\sss\tau\geqslant0}g^{\tau}_{\sss\scr{H}}
(L_0)=\Ccup_{\sss t\geqslant0}g^{t}_{\sss H}(L_0),
$$
where
\begin{gather}
\notag L_0=\big\{\mathbf{x}=(-\xi,0),\quad \mathbf{p}=
\frac{\omega}{C(\xi)}(\cos\psi,\sin\psi)\big\}\\
\intertext{and}
\label{Lambda_+}g^{t}_{H}(L_0)=
\Big(\varrho\cos\varphi\,,\;\varrho\sin\varphi\,,\;
\frac{\omega}{C(\xi)}u\cos\varphi-
\frac{\omega}{C(\xi)\varrho}v\sin\varphi\,,\;
\frac{\omega}{C(\xi)}u\sin\varphi+\frac{\omega}{C(\xi)\varrho}v
\cos\varphi\Big)
\end{gather}
 is a shift of $L_0$ along the trajectories of the vector
field $V_H$, and dependencies $\varrho(\bop,\psi)$,
$\varphi(\bop,\psi)$, $u(\bop,\psi)$, $v(\bop,\psi)$ are defined
by---again the elliptic---expressions \eqref{SOL}.
\begin{figure}[h!]
\begin{minipage}[h]{0.59\linewidth}
\includegraphics[width=1\linewidth,height=5cm]{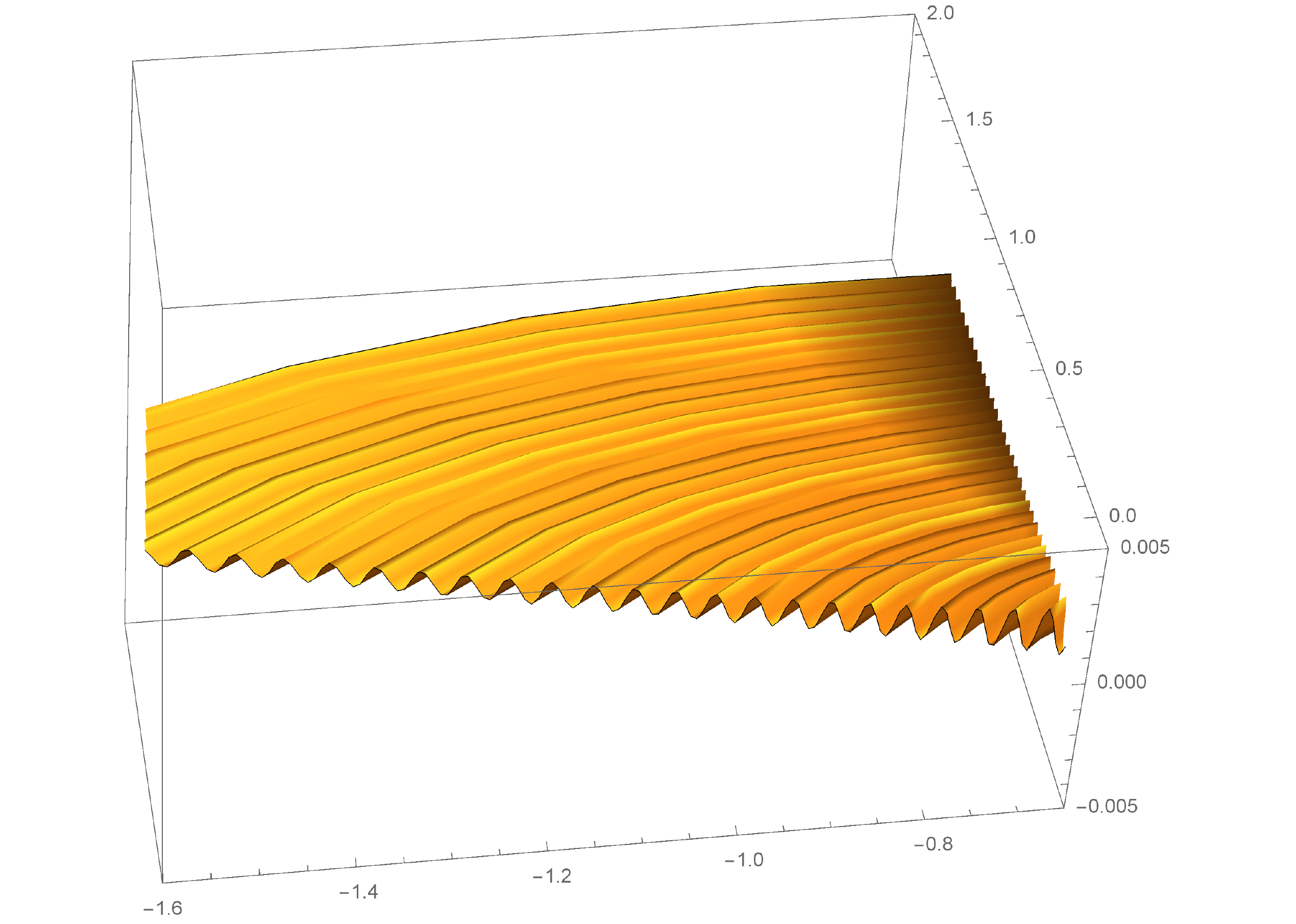}
\caption{Asymptotic solution of the equation \eqref{in_WEq4}
in the neighborhood of nonfocal points for $a=100, b=1, \xi=\frac12,
c=\frac{1}{200}, d=\frac{1}{300}, \varepsilon =\frac{1}{20}$}
\label{ex2_1}
\end{minipage}
\hfill
\begin{minipage}[h]{0.39\linewidth}
\center{\includegraphics[width=1\linewidth,height=5cm]{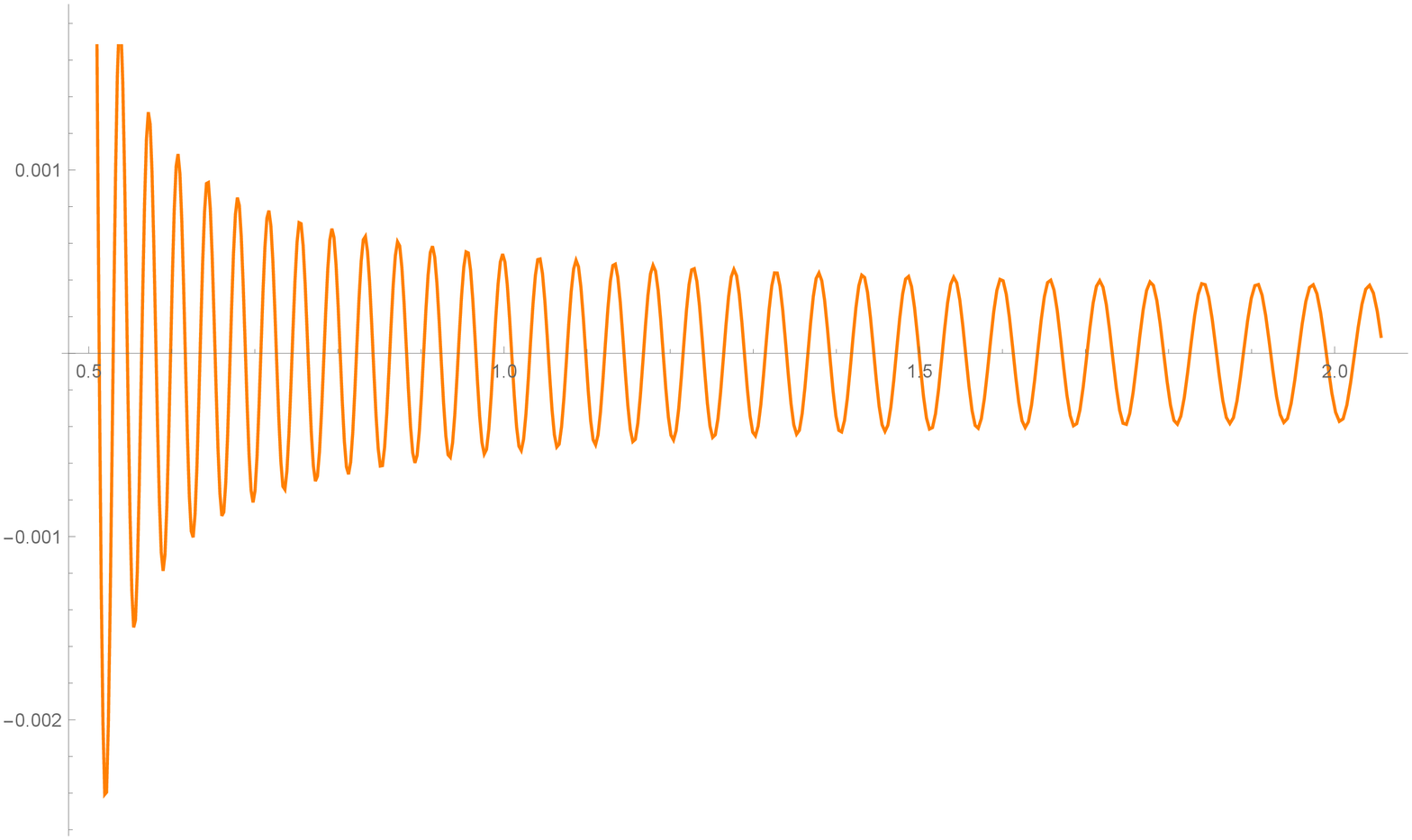}}
\caption{Asym\-pto\-tic solution in the direction
$\psi=\frac{9}{10}\pi$} \label{ex2_2}
\end{minipage}
\end{figure}

Every point on the manifold $\Lambda_{+}$ is defined by two
coordinates $(\psi, t)$, where $\psi$ is the coordinate on $L_0$ and
the coordinate $t$ characterizes time along the trajectory of the
field $V_{H}$. Note that the coordinate $t$ can be changed to the
coordinate $\tilde{\bop}=\bop-\po$. Using the equalities
$$
\int_{0}^{t} p \cdot
H_{p}\Big|_{\begin{smallmatrix}\sss\mathbf{x}=\mathbf{X}(\eta,\psi)\\
\sss\mathbf{p}=\mathbf{P}(\eta,\psi)
\end{smallmatrix}}d\eta=\omega\,t\,, \quad \frac{\pa \bop}{\pa \tau} =
\frac{2\omega}{\delta+\wp(\bop)}\,,
$$
one can derive the following result.

\begin{proposition}
The equation \eqref{in_WEq4} with the right-hand side \eqref{V} has
an asymptotic solution $v(\mathbf{x},\varepsilon)$, which satisfies
the asymptotic limiting absorption principle. If preimage of the
point $\mathbf{x}$ in $\Lambda_{+}$ is a unique and nonsingular
point $(\bop,\psi)$, then the principal part of the asymptotic
solution $v(\mathbf{x},\varepsilon)$ near $\mathbf{x}$ can be
written in the form
\begin{align*}
v(\mathbf{x},\varepsilon)&=cd\,\frac{C(\varrho)}{C(\xi)}
\frac{\sqrt{\pi
\big|\delta(\psi)+\wp(\tilde{\bop}+\po)\big|}}
{\sqrt{\varepsilon \omega|\tilde{J}(\tilde{\bop},\psi)|}}\,
\exp\left[-\frac{\omega^2(c^2\cos^2\psi+d^2\sin^2\psi)}{2C^2(\xi)}
\right]\\
&\quad\times\exp\left[\frac{i\omega}{\varepsilon}\big(\delta(\psi)
\tilde{\bop}-\zeta(\tilde{\bop}+\po)+
\zeta(\po)\big)-\frac{i\pi}{2}\left(\frac{1}{2}+
\mathrm{ind}\,\gamma\right)\right],
\end{align*}
where $\delta(\psi)$ is defined by \eqref{alpha} and
$\varrho=\varrho(\tilde{\bop}+\po,\psi)$, $\varphi
=\varphi(\tilde{\bop}+\po,\psi)$, $\tilde{J}(\tilde{\bop},\psi)=
\det\big|\frac{\pa(\varrho\cos\varphi,\varrho\sin\varphi)}
{\pa(\tilde{\bop},\psi)}\big|$. The $\mathrm{ind}\,\gamma$ is a
Maslov index of the path $\gamma$ joining nonsingular points
$(+0,\psi_0)$ and $(\tilde{\bop},\psi)$ on~$\Lambda_+$.
\end{proposition}
The latter formula for asymptotics is illustrated in the
Figs.~\ref{ex2_1} and \ref{ex2_2}.

In conclusion, note that the solutions of the considered Hamiltonian
systems allow us to derive an analytical formula for the asymptotic
solution of the Helmholtz equation with a localized right-hand side
in a neighborhood of a regular point. In the general case---say, in
a neighborhood of caustics---the solution can be represented in the
form of the Maslov canonical operator on the Lagrangian manifold
$\Lambda_+$. This manifold is defined by the expression
\eqref{Lambda_+} and involves the exact solutions of the Hamiltonian
system.

\section*{Acknowledgments}

The authors wish to express their gratitude to S.~Dobrokhotov and
M.~Pavlov for useful discussions.

The first part of the research was supported by Tomsk State
University within the framework of the national project
``Priority-2030''. The research in the second section was supported by
the Federal Target Program (AAAA-A20-120011690131-7).

\medskip

\noindent
\parbox{7cm}{Yu.~V.~Brezhnev\\
Tomsk State University\\
Tomsk, 634050 Russia\\
\texttt{brezhnev@phys.tsu.ru}} \hfill
\parbox{9cm}{A.~V.~Tsvetkova\\
Ishlinsky Institute for Problems in Mechanics RAS\\
Moscow, 119526 Russia\\
\texttt{annatsvetkova25@gmail.com} }
\end{document}